\documentclass[a4paper,12pt]{article}

\usepackage{amssymb}
\usepackage{amsthm}
\usepackage{graphicx}
\usepackage{amsmath}
\usepackage{dsfont}
\usepackage{xcolor} 

\usepackage{pstricks}
\usepackage{pstricks-add}
\usepackage{multido}
\usepackage{pst-fill}
\usepackage{pifont}

\theoremstyle{plain}
\newtheorem{theorem}{Theorem}
\newtheorem{lemma}[theorem]{Lemma}

\theoremstyle{definition}
\newtheorem{definition}[theorem]{Definition}

\theoremstyle{remark}
\newtheorem{remark}[theorem]{Remark}

\numberwithin{equation}{section}
\numberwithin{theorem}{section}

\newcommand{\X}{\vec X}
\newcommand{\x}{\mathbf x}
\newcommand{\Xs}{\mathbf X}
\renewcommand{\P}{\mathbb P}
\newcommand{\E}{\mathbb E}
\newcommand{\argmax}{\arg\max}
\newcommand{\pen}{\text{pen}}
\newcommand{\p}{\text{p}}
\newcommand{\wQ}{\widetilde Q}
\newcommand{\wP}{\widehat\P}
\newcommand{\M}{\mathcal M}
\newcommand{\rs}{r\!:\!s}
\newcommand{\1}{\mathbf 1}

\usepackage{vmargin}
\setpapersize{A4}
\setmargins{3cm}{2cm}{15.5cm}{22cm}{1cm}{0.5cm}{1cm}{2cm}

\usepackage{sectsty}
\sectionfont{\large}

\begin{document}

\title{A model selection approach for multiple sequence segmentation and dimensionality 
reduction}

\author{Bruno M. de Castro and Florencia Leonardi}

\date{\today}

\maketitle

\begin{abstract}
In this paper we consider the problem of segmenting $n$ aligned random sequences of equal length $m$ 
 into a finite number of independent 
blocks. We propose to use a penalized maximum likelihood criterion to infer 
simultaneously the number of points of independence as well as the position of each one 
of these points. We show how to compute the estimator  efficiently by means of a 
dynamic programming algorithm with  time complexity $O(m^2n)$. We also propose 
another algorithm, called \emph{hierarchical algorithm},  that provides an approximation to the  estimator when the sample size increases and runs in time  $O(mn)$. 
Our main theoretical result is the proof of almost sure consistency of the estimator and 
the convergence of the hierarchical algorithm when the sample size $n$ grows to infinity. 
We illustrate the convergence of these algorithms through some simulation examples and 
we apply the method to real protein sequence alignments of Ebola Virus.
\end{abstract}

\section{Introduction}

The problem of multiple sequence segmentation and dimensionality reduction appears in 
many application areas, including the analysis of multiple alignments of DNA/RNA and 
Amino Acid (AA) sequences. These datasets are characterized by the alignment of 
sequences belonging to the same group, as genes, protein families, etc. In general the 
length of these alignments is big and direct multivariate analysis is not possible. On the 
other hand, the identification of relevant patterns in these datasets is crucial for many 
biological problems, as for example the identification of specific patterns associated to 
some phenotypic variables.   An incomplete list of methods for sequence segmentation 
and pattern identification of biological sequences include \cite{bejeranoetal2001,gwadera2006,meme2009}.

In this paper we consider the problem of segmenting  $n$ aligned random sequences of equal length $m$  into a finite number of independent 
blocks. To be more precise, let $\X= (X_1,\dotsc,X_N)$ be a random vector with probability distribution $\P$ taking values in $A_1\times\dotsc\times A_N$, where $A_i$ is a finite alphabet for all $i=1,\dotsc, N$. 
We say $j=i+\frac12$, with $i\in\{1,2,\dotsc, N-1\}$,  is a point of independence for $\X$ if the random vectors  $(X_1,\dotsc,X_j)$ and $(X_{j+1},\dotsc,X_N)$ are independent.  We are interested in inferring the maximal set of points of independence of a random vector, when the number of such points is unknown. This is a model selection problem because the different parameter sets of the probability distributions belong to vector spaces with different dimensions. We propose a penalized maximum likelihood criterion to infer simultaneously the number of points of independence and  their positions. Our main theoretical result is the proof of the almost sure convergence of the estimator to the true set of points of  independence when the sample size $n$ increases.  

From the computational point of view we show how to compute the estimator by means of a dynamic programming algorithm that runs in time $O(m^2n)$. 
When the size of the sequences is large this can be time consuming. In these situations we propose a suboptimal algorithm running in time $O(mn)$  that also converges almost surely to the set of points of independence when the sample size $n$ increases. 

To our knowledge this is the first model selection  approach for multiple sequence segmentation based on the assumption of independence between segments. 
 The change point detection problem is a well studied phenomenon in the area of time series analysis (see for example  \cite{fryzlewicz2014} and references therein) but in general the setting and the goals are different from ours. 
 
In the next section we present the main definitions and state the main theoretical result. In Section~\ref{secalgs} we describe the algorithm to compute exactly the estimator and we present the \emph{hierarchical algorithm} that provides an efficient but suboptimal solution. 
 In Section~\ref{secmain} we prove the theoretical results and in Sections~\ref{sims} and \ref{data} we present the simulations and results on real data, respectively.   
 
\section{Likelihood function and model selection}\label{secdefs}

Given two integers $r\leq s$ denote by $\rs$ the integer interval $r,r+1,\dotsc,s$. 
Denote by $J_{r:s}$ the set $\{i + \frac{1}{2} \colon i=r,1,\dotsc,s-1\}$.
We say $U_{r:s}\subset J_{r:s}$ is a \emph{maximal} set of points of independence for the interval $\rs$ if no $v\in J_{r:s}\setminus U_{r:s}$ is a point of independence for $\X$.  For each random vector $\X$ and each interval $\rs$ there is only one  maximal set of points of independence; from now on this special set will be denoted by $U^*_{r:s}$. In the special case $r=1, s=N$ we will simply write 
$U$, $J$ and $U^*$.

Without loss of generality we will also suppose that the set $U_{r:s}\subset J_{r:s}$ is ordered; in this case $U_{r:s}=(u_1,\dotsc, u_k)$ with $u_i<u_j$ if $i<j$. 
From $U_{r:s}$ it is possible to obtain the set of blocks of independent variables as the set $B(U_{r:s}) = \{I_1,I_2,\dotsc,I_{k+1}\}$ of integer intervals given by
\begin{align*}
I_1 = r:(u_1-{\textstyle\frac12})\,,\qquad
I_{i} = (u_{i-1}+{\textstyle\frac12}):(u_{i}-{\textstyle\frac12})\,,\quad i=2,\dotsc,k
\end{align*}
and
\[
I_{k+1} = (u_k+{\textstyle\frac12}):s\,.
\]
Given an integer interval $I=\rs$ denote by $A^I$ the set of finite sequences $A_r\times\cdots\times A_s$ and by $|A^I| = \prod_{i=r}^s |A_i|$, with $|A_i|$ the cardinal of the set $A_i$.

In this way, given a vector of observations $\x = (x_1,\dotsc,x_N)$, the likelihood function for the set $U$ can be written as 
\[
L(U;\x) \;=\; \prod_{I\in B(U)} \P(X_j=x_j\colon j\in I)\,.
\]

Assume we observe an i.i.d sample $\x^{(1)},\dotsc,\x^{(n)}$ of $\X$, denoted by $\Xs$.    Then we have 
\[
L(U;\Xs) \;=\; \prod_{i=1}^n\prod_{I\in B(U)} \P(X_j=x^{(i)}_j\colon j\in I)\,.
\]
Denote by $X_{r:s}$ the sequence $X_r,\dotsc,X_s$ and by $\Xs_{r:s}$ the i.i.d sample  
$\x^{(1)}_{r:s},\dotsc,\x_{r:s}^{(n)}$.
Given a finite sequence $a_{r:s}\in A^{r:s}$ define
\[
N(a_{r:s})\;=\; \sum_{i=1}^n \1\{x^{(i)}_{r:s}=a_{r:s}\}\,.
\]
Then $L(U;\Xs)$ can be rewritten as
\begin{equation}\label{like}
L(U;\Xs) \;=\; \prod_{I\in B(U)}\prod_{a_I\in A^{I}} \P(X_{I}=a_{I})^{N(a_{I})}\,.
\end{equation}

To proceed we need a model for the block $I$ and some estimators for the probabilities 
$\P(X_{I}=a_{I})$. 
Assume that for each block $I=\rs\subset 1\!:\!N$ we have a model assignment $\M(I)$ 
and a   penalization function 
$\p(\Xs_I)$ satisfying 
\begin{itemize}
\item[\bf (A1)] For any $i\in r\!:\!(s-1)$ we have 
\[
\p(\Xs_{r:s}) \; > \; \p(\Xs_{r:i}) + \p(\Xs_{(i+1):s})\,.
\]
\end{itemize}

We will assume that in model $\M(I)$ we have maximum likelihood estimators for the probabilities $\P(a_{I})$; that is some values $\widehat \P(a_{I})$ maximizing \eqref{like}. Then we can plug-in these estimators in \eqref{like} obtaining  
\[
\hat L(U;\Xs) \;=\; \prod_{I\in B(U)}\prod_{a_I\in A^{I}} \widehat \P(a_{I})^{N(a_{I})}\,.
\]

\begin{remark}
As an example we could consider for any block $I$ the multivariate distribution over $A^I$ as possible model assignment $\M(I)$. For this class of models we have that the maximum likelihood estimators are given by
\[
\wP(a_I) \;=\; \frac{N(a_I)}{n}\,,\qquad a_I\in A^I\,.
\]
In this case, for example, we could take as penalizing function 
$\p(\Xs_I) = c(|A^I|-1)$, the number of parameters in the multivariate distribution over
 $A^I$ (in case  $A^I$ is known), multiplied by a positive constant $c$.  
 In the case $A^I$ is unknown a priori, it could be appropriate to use an empirical version of  $\p(I)$; namely  
 $\hat\p(\Xs_I) = c(\prod_{i\in I}\|\Xs_{i}\|_0-1)$, where $\|\Xs_{i}\|_0 = \sum_{a_i\in A_i} \1\{N(a_i)>0\}$
 and $c$ is a positive constant.
In some cases other model assignments or penalizing functions could be more efficient or  parsimonious, mainly in the case of large 
blocks $I$. One example is to consider a Markov chain of some appropriate order. 
\end{remark}

%

In the sequel it will be easier to work with the logarithm of the estimated likelihood function, for this reason we define
\[
\hat\ell(U;\Xs) \;=\; \sum_{I\in B(U)}Q(I,\Xs)\,,
\]
where 
\[
Q(I,\Xs) = \sum_{a_I\in A^{I}} N(a_{I}) \log \widehat \P(a_{I})\,.
\]

Now we introduce the model selection criterion based on the
maximization of the penalized log-likelihood.

\begin{definition}\label{defuhat}
Given a sample $\Xs$ define
\begin{equation}\label{uhat}
\hat U(\Xs) \;=\; \underset{U\subset J, U \text{ordered}}{\arg\max}\{ \;\hat\ell(U,\Xs) - \pen(U,n)\;\}\,,
\end{equation}
where
\[
\pen(U,n) = \sum_{I\in B(U)} \p(\Xs_I)\log(n)\,.
\]
\end{definition}

The first theoretical result of this article shows that when $n$ grows to infinity, the estimator \eqref{defuhat} is equal to the true set $U^*$ eventually almost surely. In other words, with probability one there exists $n_0$ (depending on the entire sample $\x^{(1)},\dotsc,\x^{(n)}$) such that for all $n\geq n_0$ we have $\hat U(\Xs)=U^*$.

\begin{theorem}\label{main}
The estimator \eqref{defuhat} is strongly consistent; that is 
\[
\hat U(\Xs) = U^*
\]
eventually almost surely when $n\to\infty$.
\end{theorem}

The proof of Theorem~\ref{main} is postponed to Section~\ref{secmain}.

\section{Computation of the independence set estimator}\label{secalgs}

In this section we show how to compute efficiently the penalized maximum likelihood estimator given in Definition~\ref{defuhat}.  
The first subsection, mostly inspired in \cite{hawkins1976},  presents a dynamic programming algorithm that computes 
 exactly the optimal argument of \eqref{uhat},  at the cost of a quadratic computing time in terms of  the length of the sequences (see also \cite{bai-perron2003} and references therein).  In the second subsection we propose a \emph{divide and conquer} approximation for the optimum in \eqref{uhat}, at a more efficient computing time.  
 We show that this second algorithm also retrieves the true set of points of independence with probability one when the sample size grows.  

\subsection{Dynamic programming algorithm}


Let $F_{k+1}(N)$ denote the maximum value of the function in  (\ref{uhat}) corresponding to a $k$-dimensional vector $U$ for the sample $\Xs$; that is
\[
F_{k+1}(N) = \underset{U,\; |U|=k}{\max}\{\,\hat\ell(U,\Xs) - \text{pen}(U,n)\, \}\,.
\]

 It is easy to see that the optimal $k$-dimensional vector $U$
leading to $F_{k+1}(N)$ consist of $k-1$ independence points over 
$1\!:\!i$ and a single block $(i+1)\!:\!N$, where $i+1/2$ is the rightmost  point
of independence. 
Moreover, the $k$ blocks over $1\!:\!i$ must maximize the function (\ref{uhat}) for the sample $\Xs_{1:i}$, attaining $F_{k}(i)$. 
In this way, the dynamic programming recursion is
\[
F_1(N) = \widetilde Q(1\!:\!N,\Xs)\,,\qquad F_{k+1}(N) = \max_{i=k,\dotsc,N-1}\{\,F_{k}(i)+ \widetilde Q((i+1)\!:\!N,\Xs)\,\}\,,
\]
where $\widetilde Q$ is given by
\[
\widetilde Q(I,\Xs) = Q(I,\Xs) - \p(\Xs_I)\log(n)\,.
\]

The estimator $\hat U(\Xs)$ in Definition~\ref{defuhat} is computed by tabulating 
$F_1(i)$ for all $i$ up to $N$, and then by computing $F_2(i)$ for all $i$ 
and so on up to $F_{N}(N)$. The optimal value of $k$ is obtained by
the equation
\[
\hat k = \underset{k=1,\dots,N}{\argmax}\{\, F_k(N)\,\} - 1\,.
\]
and the vector $\hat U(\Xs)=(\hat u_1,\dotsc, \hat u_{\hat k})$ is given by
\begin{align*}
\hat u_{\hat k} &= \underset{i=\hat k,\dotsc,N-1}{\argmax}\{\, F_{\hat k}(i)+ \widetilde Q((i+1)\!:\!N,\Xs)\,\} + \frac12\,,\\
\hat u_{i}& = \underset{\ell=i,\dotsc,\hat u_{i+1}}{\argmax}\{\, F_i(\ell)+ \widetilde Q((\ell+1)\!:\!N,\Xs)\,\} + \frac12\,, \quad i = 1,\dotsc,\hat k-1\,.
\end{align*}

\subsection{Hierarchical Algorithm}

Here we present an efficient  \emph{divide and conquer} algorithm to compute the estimator given by Definition~\ref{defuhat} with computational cost $O(nm)$.

Let $I= \rs$ be an integer interval (with the convention that $I=\emptyset$ if $s<r$). Define
\begin{equation}\label{h}
h(I,\Xs)=\underset{i\in I}{\argmax} \{\, \wQ(r\!:\!(i-1),\Xs)+\wQ(i\!:\!s,\Xs)\,\} - r -   \frac12\,,
\end{equation}
where also by convention we have $\wQ(\emptyset,\Xs) = 0$. 
%

\begin{figure}[t!]
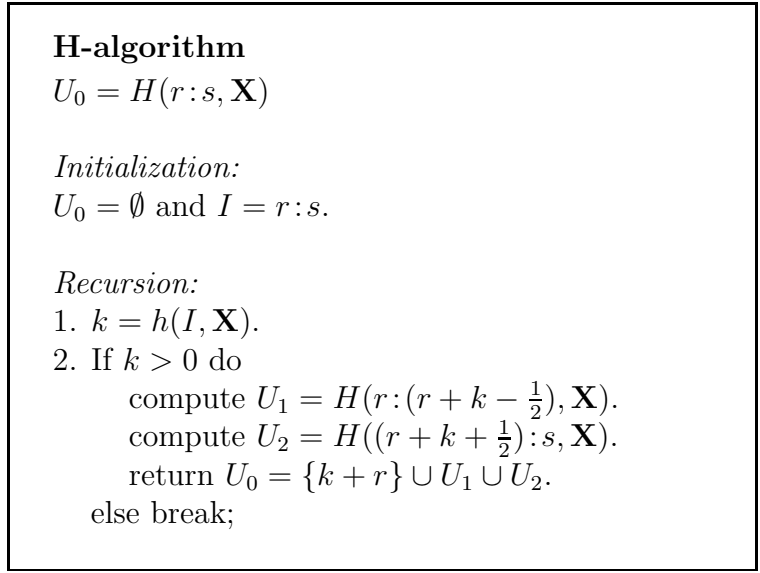

\begin{center}
\fbox{
\hspace*{2mm}
\begin{minipage}{9cm}
\vspace*{3mm}
\noindent{\bf H-algorithm\\[2pt]
$U_0 = H(\rs,\Xs)$}\\

\noindent{\em Initialization:} \\
$U_0=\emptyset$  and $I=\rs$. \\

\noindent{\em Recursion:} \\
1. $k = h(I,\Xs)$.\\
2. If $k>0$ do\\
\hspace*{10mm}compute $U_1 = H(r\!:\!(r+k-\frac12),\Xs)$.\\
\hspace*{10mm}compute $U_2 = H((r+k+\frac12)\!:\!s,\Xs)$.\\
\hspace*{10mm}return $U_0  = \{k+r\}\cup U_1\cup U_2$.\\
\hspace*{5mm}else break;\\
\end{minipage}
}
\end{center}
\caption{$H$-algorithm.}\label{halg}
\end{figure}

The idea of this algorithm is to compute the best point of  independence for the interval $I$ (if there is one point in the interval leading to a maximum of the penalized likelihood or a point outside $I$ otherwise) and then to iterate this criterion on both segments separated by this point, until no more points are detected. A summary of the Hierarchical algorithm is given in Figure~\ref{halg}. 

We show the almost sure convergence of the hierarchical algorithm in the following theorem. 

\begin{theorem}\label{thalg}
$H(1\!:\!N,\Xs) = U^*$ eventually almost surely as  
$n\to\infty$.
\end{theorem}

The proof of this theorem is given in the following section.

\section{Proofs}\label{secmain}

Let $I\subset 1\!:\!N$  be an integer interval. 
Given two probability distributions $P_1$ and $P_2$ over $A^I$ let  $D(P_1\|P_2)$ denote the \emph{Kullback-Leibler} divergence between $P_1$ and $P_2$;  that is
\begin{align*}
D(P_1\|P_2)\;&=\; \sum_{a_I:P_1(a_I)>0} P_1(a_I) 
\log\Bigl(\,\frac{P_1(a_I)}{P_2(a_I)}\, \Bigr)\,.
\end{align*}
A well-know property about the Kullback-Leibler divergence between two probability distributions states that $D=0$ if and only if $P_1=P_2$. 


\begin{lemma}\label{lil}
For any block $I\subset 1\!:\!N$ and any $\alpha>0$ we have that  
\[
D(\wP\,\|\,\P)\;< \;\frac{\alpha\log(n)}{n}
\]
eventually almost surely as $n\to\infty$, where $\wP$ and $\P$ represent the marginal distributions over 
the interval $I$. 
\end{lemma}
\begin{proof}
Define, for a fixed $a_I\in A^I$ with $\P(a_I)>0$, the random variables 
\[
Y_i(a_I) = \mathbf{1}\{ \x_I^{(i)}=a_I\} - \P(a_I)\,, \; i=1,2,\ldots,n
\]
and
\[
Z_n(a_I)=\sum_{i=1}^n Y_i(a_I) = N(a_I) - n\P(a_I).
\]
The variables $\{Y_i(a_I) : i=1,2,\ldots,n\}$ are independent and identically distributed, with $\E(Y_i(a_I))=0$ and $\E(Y_i(a_I)^2)=\P(a_I)(1-\P(a_I))$. Then, by the Law of the Iterated Logarithm (cf. Theorem~3.52 in \cite{breiman1992}) we have that for any $\epsilon > 0$
\begin{equation}
|Z_n(a_I)| \,<\, (1+\epsilon)\,\P(a_I)(1-\P(a_I))\,\sqrt{2n\log \log n} 
 \end{equation}
eventually almost surely as $n\to \infty$. 
Dividing both sides of the 
last inequality by $n\sqrt{\P(a_I)}$ we obtain that for any $\delta>0$
\begin{equation}\label{ineqp}
\frac{|Z_n(a_I)/n|}{\sqrt{\P(a_I)}}\,=\,\frac{|\wP(a_I)-\P(a_I)|}{\sqrt{|\P(a_I)}} \,<\, \sqrt{\frac{\delta\log(n)}{n}}
\end{equation}
eventually almost surely as $n\to \infty$. 
Now by Lemma~6.3 in \cite{csiszar2006} and \eqref{ineqp}
 we have that 
\begin{align*}
D(\wP\|\P) \;&\leq \,\sum_{a_I: N(a_I)>0}\; \frac{ |\wP(a_I)-\P(a_I)|^2}{\P(a_I)}\\
&\leq \,|A^I| \; \max_{a_I:\P(a_I)>0}\; \frac{ |\wP(a_I)-\P(a_I)|^2}{\P(a_I)}\\
&\leq\;\frac{\delta\log(n)}{|A^I|n}\,
\end{align*}
eventually almost surely as $n\to \infty$ and this implies the result for $|A^I|$
finite. 
\end{proof}

For any $j\in I=\rs$  denote by $\widetilde\P_{j}$ the probability distribution  given by 
\[
\widetilde\P_{j}(a_I) \;=\; \P(a_{r:j}) \P(a_{(j+1):s}),\qquad a_I\in A^{I}.
\]

\begin{lemma}\label{lem1}
Let $I=\rs\subset 1\!:\!N$ and suppose $I\subset I^*$, with $I^*\in B(U^*)$; that is, there is no point of independence in the rage of $I$. Then there exists a constant $\delta>0$
such that for any $i\in I, i\neq s$ we have 
\[
Q(I;\Xs) \,-\, Q(r\!:\!i;\Xs)  \,-\, Q((i+1)\!:\!s;\Xs) \;>\;  \delta \,n
\]
eventually almost surely when $n\to \infty$. 
\end{lemma}

\begin{proof}
Note that for any $a_r^i\in A^{r:i}$ we have $N(a_r^i) = \sum_{a_{i+1}^s}N(a_r^i a_{i+1}^s)$
and analogously for any $a_{i+1}^s\in A^{(j+1):s}$, $N(a_{i+1}^s) = \sum_{a_r^i}N(a_r^i a_{i+1}^s)$. Then 
\[
Q(I;\Xs) \,-\, Q(r\!:\!i;\Xs)  \,-\, Q((i+1)\!:\!s;\Xs) =  \sum_{a_r^s\in A^{I}} N(a_r^s) 
\log\Biggl(\frac{\wP(a_r^s)}{\wP(a_r^i)\wP(a_{i+1}^s)} \Biggr)
\]
Dividing by $n$ and taking limit when $n\to\infty$ we have that  the expression above converges almost surely to
\[
 \sum_{a_r^s\in A^{I}} \P(a_r^s) 
\log\Biggl(\frac{\P(a_r^s)}{\P(a_r^i)\P(a_{i+1}^s)} \Biggr)\;=\; D(\P\|\widetilde\P_{j})\;>\;0\,.
\]
The last inequality follows because $i+\frac12$ is not a point of independence for $\vec{X}$. To conclude the proof, it is enough to take 
\[
\delta = \min_{j\in I,j\neq s}{D(\P\|\widetilde\P_{j})}/2\,. \qedhere
\]
\end{proof}

\begin{lemma}\label{lem2}
Let $I=\rs\subset 1\!:\!N$ and suppose there exists $i\in I$ such that $i+\frac12 \in U^*$; that is, there is a point of independence in the rage of $I$. Then for any $\alpha>0$  we have 
\[
Q(I;\Xs) \,-\, Q(r\!:\!i;\Xs)  \,-\, Q((i+1)\!:\!s;\Xs) \;<\;  \alpha \log(n)
\]
eventually almost surely when $n\to \infty$. 
\end{lemma}

\begin{proof}
As in the proof of Lemma~\ref{lem1} we can write 
\[
Q(I;\Xs) \,-\, Q(r\!:\!i;\Xs)  \,-\, Q((i+1)\!:\!s;\Xs) =  \sum_{a_r^s\in A^{I}} N(a_r^s) 
\log\biggl[\frac{\wP(a_r^s)}{\wP(a_r^i)\wP(a_{i+1}^s)} \biggr]\,.
\]
As $\hat\P(\cdot)$ is the maximum likelihood estimator of $\P(\cdot)$ and $i+\frac12$ is a point of independence  we have that
\begin{align*}
 Q(r\!:\!i;\Xs)  \,+\, Q((i+1)\!:\!s;\Xs)  \;&\geq\;  \sum_{a_r^s\in A^{I}} N(a_r^s) 
\log\bigl[\P(a_r^i)\P(a_{i+1}^s) \bigr]\\
&=\;  \sum_{a_r^s\in A^{I}} N(a_r^s) 
\log\bigl[\P(a_r^s) \bigr]
\end{align*}
Then by combining this last inequality and Lemma~\ref{lil} we have that 
\begin{align*}
Q(I;\Xs) \,-\, Q(r\!:\!i;\Xs)  \,-\, Q((i+1)\!:\!s;\Xs) \;&\leq\; \sum_{a_r^s\in A^{I}} N(a_r^s) 
\log\biggl[\frac{\wP(a_r^s)}{\P(a_r^s)} \biggr]\\
&< \alpha \log(n)
\end{align*}
eventually almost surely as $n\to\infty$. 
\end{proof}

\begin{proof}[Proof of Theorem~\ref{main}]
We will show that if $U\subset J$ (ordered) is such that $U\neq U^*$ then there exists a modification $U'$ of $U$ such that
\begin{equation}\label{ineq}
\hat\ell(U';\Xs) - \pen(U',n) \;> \; \hat\ell(U;\Xs) - \pen(U,n)\,
\end{equation}
eventually almost surely as $n\to\infty$. 
As the number of possible such modifications of $U$ leading to $U^*$ is finite this implies the result. 
Consider the following cases
\begin{itemize}
\item[(a)] $U\subset U^*$
\item[(b)] $U\not\subset U^*$. 
\end{itemize}
First we will prove \eqref{ineq} for case (a). Note that in this case it is enough to show it in the case $U^*\neq\emptyset$. Assume $U^*=(u^*_1,\dotsc, u^*_K)$, with $K\geq 1$,  $U=(u_1,\dotsc, u_k)$ and let $i$ be the first index such that $u_i^*\notin U$. Define $U'= U\cup\{u_i^*\} = (u_1,\dotsc,u_{i-1},u_i^*,u_{i+1},\dotsc,u_k)$ (see Figure~\ref{fig2}).

\begin{figure}[t]
\begin{center}
\psset{unit=1cm}
\begin{pspicture}(9.5,8)(-0.5,2)
\large
\psset{linewidth=2pt,arrowsize=12pt}
\psline(1,7)(9,7)\rput[l]{0}(0,7){$U^*$}
\psline(1.3,6.8)(1.3,7.2)\rput[b]{0}(1.3,7.5){$u_{i-1}^*$}
\psline(4,6.8)(4,7.2)\rput[b]{0}(4,7.5){$u_i^*$}
\psline(6.5,6.8)(6.5,7.2)\rput[b]{0}(6.5,7.5){$u_{i+1}^*$}
\psline(8.7,6.8)(8.7,7.2)\rput[b]{0}(8.7,7.5){$u_{r}^*$}
\rput[b]{0}(7.7,7.5){$\cdots$}
\rput{0}(0,-2){
\psline(1,7)(9,7)\rput[l]{0}(0,7){$U$}
\psline(1.3,6.8)(1.3,7.2)\rput[b]{0}(1.3,7.5){$u_{i-1}$}
\psline(8.7,6.8)(8.7,7.2)\rput[b]{0}(8.7,7.5){$u_{i}$}
}
\rput{0}(0,-4){
\psline(1,7)(9,7)\rput[l]{0}(0,7){$U'$}
\psline(1.3,6.8)(1.3,7.2)\rput[b]{0}(1.3,7.5){$u_{i-1}$}
\psline(4,6.8)(4,7.2)\rput[b]{0}(4,7.5){$u_i^*$}
\psline(8.7,6.8)(8.7,7.2)\rput[b]{0}(8.7,7.5){$u_i$}
}
\end{pspicture}
\caption{Case (a): $U\subset U^*$.}\label{fig2}
\end{center}
\end{figure}
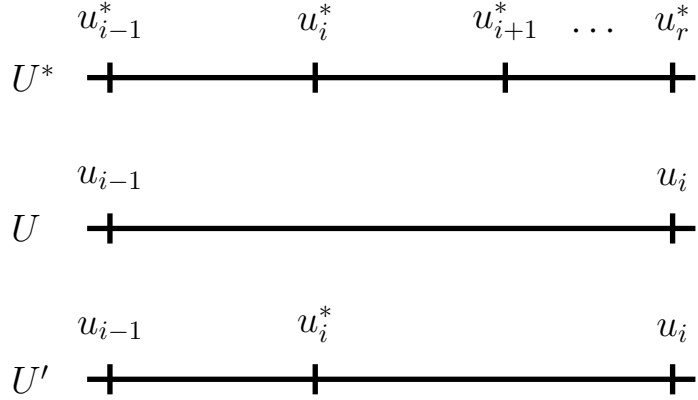
Note that $B(U')$ has the two blocks $I_1 = (u_{i-1}+\frac12)\!:\!(u_{i}^*-\frac12)$ and  
$I_2 = (u_{i}^*+\frac12)\!:\!(u_{i}-\frac12)$ replacing the single block $I = (u_{i-1}+\frac12)\!:\!(u_{i}-
\frac12)$ in $B(U)$. 
Then by Lemma~\ref{lem2} we have that for any $\alpha > 0$
\begin{align*}
\hat\ell(U';\Xs) -  \hat\ell(U;\Xs) \;&= \;Q(I_1,\Xs) + Q(I_2,\Xs) - Q(I,\Xs)\\
& >\; -\,\alpha\log(n)
\end{align*}
eventually almost surely as $n\to\infty$. 
On the other hand 
\begin{align*}
\pen(U',n) - \pen(U,n) \;&=\;  c_1 \log(n)\,,
\end{align*}
where $c_1 =   \p(\Xs_{I_1}) + \p(\Xs_{I_2})  - \p(\Xs_I)  < 0$. 
Therefore eventually almost surely as $n\to\infty$ we have that
\[
\hat\ell(U';\Xs) -  \hat\ell(U;\Xs)\;> \; \pen(U',n) - \pen(U,n)\,.
\] 

Finally we will consider case (b). As before suppose 
 $U^*=(u^*_1,\dotsc, u^*_K)$, $U=(u_1,\dotsc, u_k)$ and let $i$ be the first index such that $u_i\notin U^*$. Define $U'= U\setminus\{u_i\} = (u_1,\dotsc,u_{i-1},u_{i+1},\dotsc,u_k)$ (see Figure~\ref{fig1}).
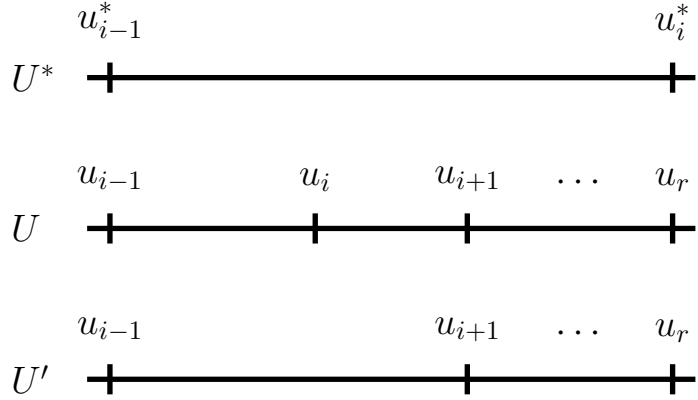
\begin{figure}[t]
\begin{center}
\psset{unit=1cm}
\begin{pspicture}(9.5,8)(-0.5,2)
\large
\psset{linewidth=2pt,arrowsize=12pt}
\psline(1,7)(9,7)\rput[l]{0}(0,7){$U^*$}
\psline(1.3,6.8)(1.3,7.2)\rput[b]{0}(1.3,7.5){$u_{i-1}^*$}
\psline(8.7,6.8)(8.7,7.2)\rput[b]{0}(8.7,7.5){$u_{i}^*$}
\rput{0}(0,-2){
\psline(1,7)(9,7)\rput[l]{0}(0,7){$U$}
\psline(1.3,6.8)(1.3,7.2)\rput[b]{0}(1.3,7.5){$u_{i-1}$}
\psline(4,6.8)(4,7.2)\rput[b]{0}(4,7.5){$u_i$}
\psline(6,6.8)(6,7.2)\rput[b]{0}(6,7.5){$u_{i+1}$}
\rput[b]{0}(7.5,7.5){$\cdots$}
\psline(8.7,6.8)(8.7,7.2)\rput[b]{0}(8.7,7.5){$u_r$}
}
\rput{0}(0,-4){
\psline(1,7)(9,7)\rput[l]{0}(0,7){$U'$}
\psline(1.3,6.8)(1.3,7.2)\rput[b]{0}(1.3,7.5){$u_{i-1}$}
\psline(6,6.8)(6,7.2)\rput[b]{0}(6,7.5){$u_{i+1}$}
\rput[b]{0}(7.5,7.5){$\cdots$}
\psline(8.7,6.8)(8.7,7.2)\rput[b]{0}(8.7,7.5){$u_r$}
}
\end{pspicture}
\caption{Case (b): $U\not\subset U^*$.}\label{fig1}
\end{center}
\end{figure}
In this case we have that $B(U')$ has the single block
$I = (u_{i-1}+\frac12)\!:\!(u_{i+1}-\frac12)$
replacing the two blocks $I_1 = (u_{i-1}+\frac12)\!:\!(u_{i}-\frac12)$ and  
$I_2 = (u_{i}+\frac12)\!:\!(u_{i+1}-\frac12)$ in  $B(U)$. Then 
\[
\hat\ell(U';\Xs) -  \hat\ell(U;\Xs) = Q(I,\Xs) - Q(I_1,\Xs) - Q(I_2,\Xs)\,. 
\]
By Lemma~\ref{lem1}, as $u_i\notin U^*$ we have that there exists a constant $\delta>0$ such that eventually almost surely 
\[
\hat\ell(U';\Xs) -  \hat\ell(U;\Xs) \; > \delta\, n\,.
\]
On the other hand we have that 
\begin{align*}
\pen(U',n) - \pen(U,n) \;&=\;  c_2 \log(n)\,,
\end{align*}
where $c_2 =  \p(\Xs_I) - \p(\Xs_{I_1}) - \p(\Xs_{I_2}) > 0$. 
Therefore eventually almost surely as $n\to\infty$ we have that 
\[
\hat\ell(U';\Xs) -  \hat\ell(U;\Xs)\;> \; \pen(U',n) - \pen(U,n)\,.
\] 
This concludes the proof of Theorem~\ref{main}.
\end{proof}

\begin{proof}[Proof of Theorem~\ref{thalg}]
First consider the case $U^*=\emptyset$. By Lemma~\ref{lem1} we have that eventually almost surely as $n\to\infty$ the value of $i$ maximizing \eqref{h} will be $i=r$, giving $h(1\!:\!N,\Xs) = -\frac12$. Therefore $H(1\!:\!N,\Xs) = \emptyset = U^*$ eventually almost surely as $n\to\infty$. 
Now suppose there is a point of independence $u\in U^*$.  We will prove that eventually almost surely $u\in H(1\!:\!N,\Xs)$. As $u\in U^*$, for any integer $r<u$ and  any integer $s>u$ we have by Lemma~\ref{lem2} that for any $\alpha>0$
\begin{equation*}
Q(r\!:\!(u-\frac12),\Xs) + Q((u+\frac12)\!:\!s),\Xs) - Q(r\!:\!s),\Xs) 
\,>  \,-\alpha\, \log(n)
\end{equation*}
eventually almost surely as $n\to\infty$. 
Then if we take  $\alpha<\p(\Xs_{r:s}) - \p(\Xs_{r:(u-\frac12)}) - \p(\Xs_{u+\frac12:s})$ we will have that eventually,
as the interval length $s-r$ decreases, $h(\rs,\Xs) = u - r$, or  equivalently $u\in H(\rs,\Xs)\subset H(1\!:\!N,\Xs)$. 
As in any iteration of the algorithm there is an interval that contains $u$ and the length decreases, this will happen eventually almost surely as $n\to\infty$. 
\end{proof}


\section{Simulations}\label{sims}

\begin{figure}[t]
\begin{center}
\begin{pspicture}(6,6)(0,0)
\rput{0}(3,3){\includegraphics[scale=0.5,trim = 20 20 10 0]{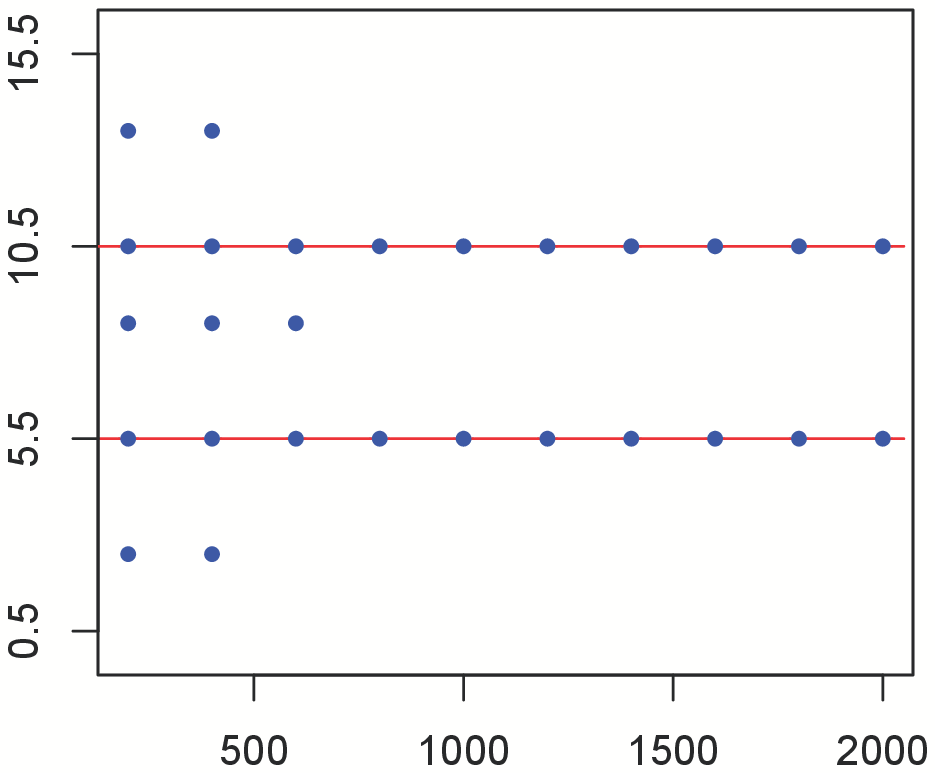}}
\rput[B]{0}(3.2,5.2){\scalebox{0.8}{$\p(\Xs_I) = |A|^{|I|}-1$}}
\rput[B]{0}(3.45,0.3){\scalebox{0.8}{$n$}}
\rput[B]{90}(0.4,2.95){\scalebox{0.7}{$\hat U(\Xs)$}}
\end{pspicture}
\begin{pspicture}(6,6)(0,0)
\rput{0}(3,3){\includegraphics[scale=0.5,trim = 20 20 10 0]{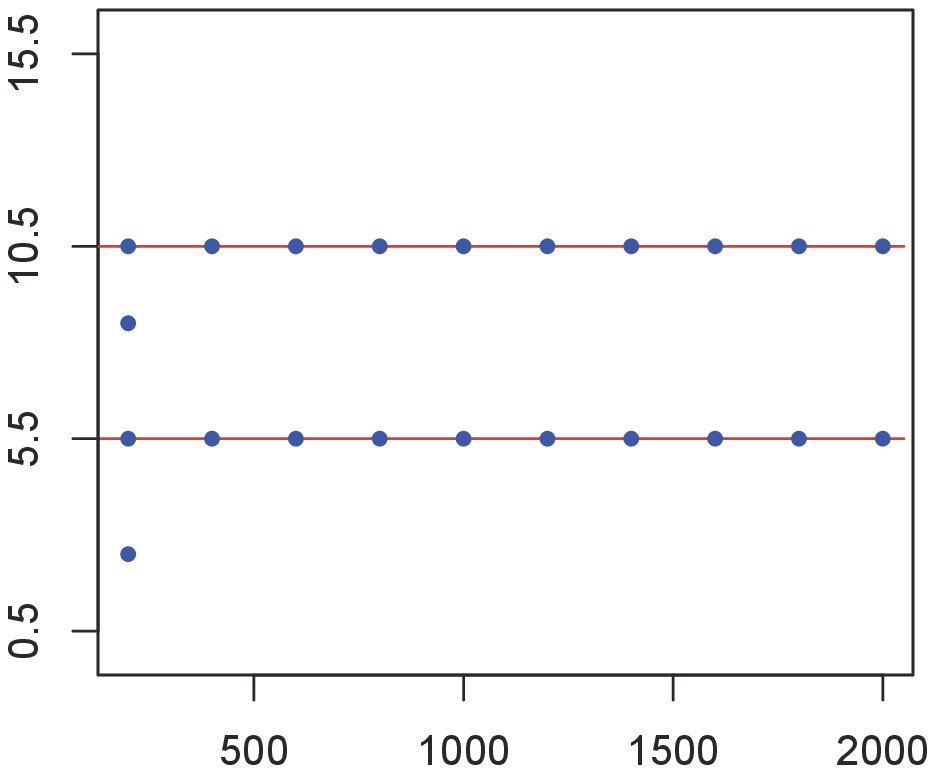}}
\rput[B]{0}(3.2,5.2){\scalebox{0.8}{$\p(I) = 0.5(|A|^{|I|}-1)$}}
\rput[B]{0}(3.45,0.3){\scalebox{0.8}{$n$}}
\rput[B]{90}(0.4,2.95){\scalebox{0.7}{$\hat U(\Xs)$}}
\end{pspicture}
\begin{pspicture}(6,6)(0,0)
\rput{0}(3,3){\includegraphics[scale=0.5,trim = 20 20 10 0]{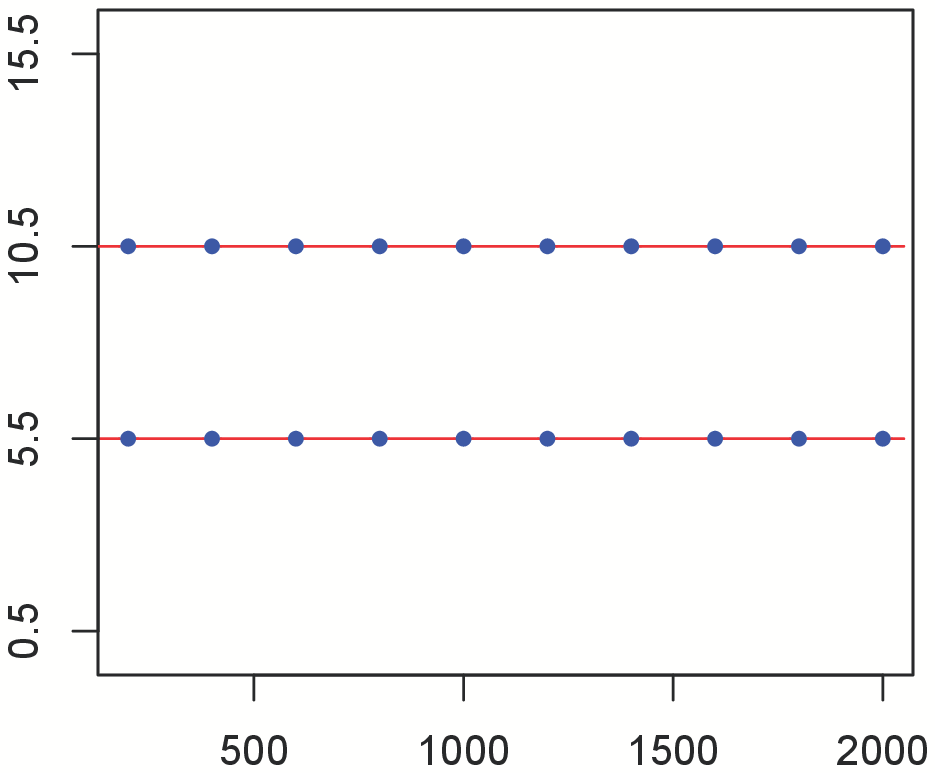}}
\rput[B]{0}(3.2,5.2){\scalebox{0.8}{$\p(I) = 0.1(|A|^{|I|}-1)$}}
\rput[B]{0}(3.45,0.3){\scalebox{0.8}{$n$}}
\rput[B]{90}(0.4,2.95){\scalebox{0.7}{$\hat U(\Xs)$}}
\end{pspicture}
\begin{pspicture}(6,6)(0,0)
\rput{0}(3,3){\includegraphics[scale=0.5,trim = 20 20 10 0]{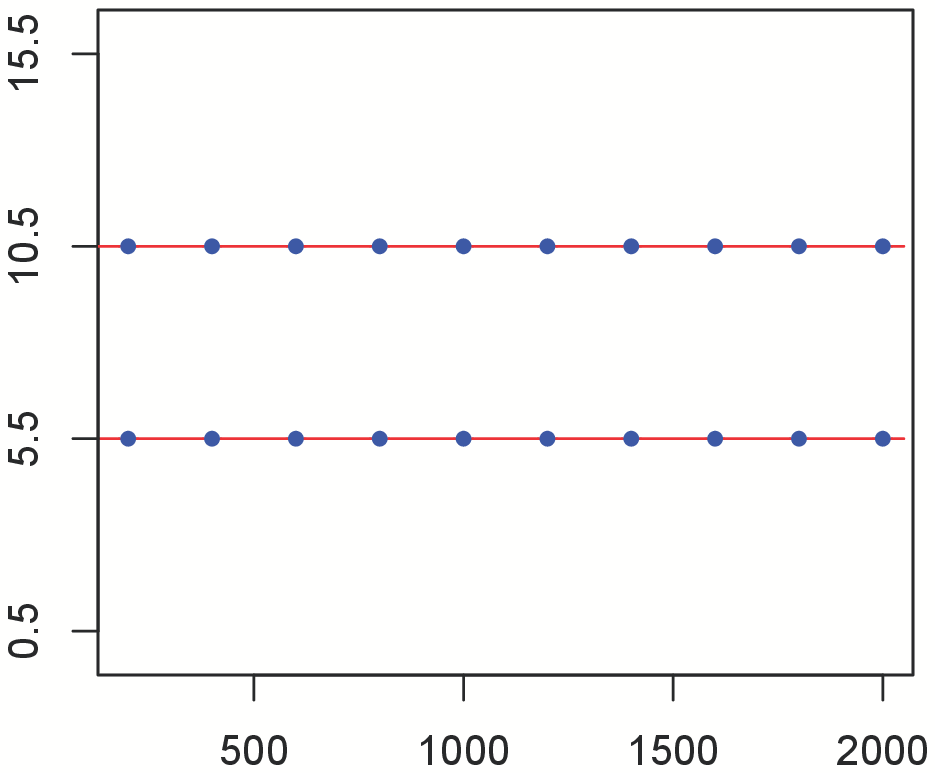}}
\rput[B]{0}(3.2,5.2){\scalebox{0.8}{$\p(I) = 0.05(|A|^{|I|}-1)$}}
\rput[B]{0}(3.45,0.3){\scalebox{0.8}{$n$}}
\rput[B]{90}(0.4,2.95){\scalebox{0.7}{$\hat U(\Xs)$}}
\end{pspicture}
\end{center}
\caption{Dynamic programming estimator of the set of points of independence $U^*=\{5.5,10.5\}$, with $|A|=2$ and for different penalizing functions.}
\label{fig1din}
\end{figure}

We perform some simulations and test both the dynamic programming (exact estimator) as well as the hierarchical algorithm. First we considered the alphabet $A=\{1,2\}$ and the transition matrices $P_1$ and $P_2$ given by
\[
P_1 = \begin{pmatrix}
\frac16 & \frac 56\\
\frac56 & \frac16
\end{pmatrix}
\,,\qquad P_2 = \begin{pmatrix}
\frac56 & \frac 16\\
\frac16 & \frac56
\end{pmatrix}\,.
\]

We simulated $n$ realizations of the vector $\vec{X} = (X_1,X_2,\dotsc,X_{15})$
such that $X_1,\dotsc,X_5$ follow a markov chain with transition matrix $P_1$, $X_6,\dotsc, X_{10}$ follow a markov chain with transition matrix $P_2$
and $X_{11},\dotsc, X_{15}$ follow again a markov chain with transition matrix $P_1$. The initial condition for any one of these markov chains is the uniform distribution over $A$. The results for different values of $n$ are given in Figure~\ref{fig1din} for the exact estimator and in  Figure~\ref{fig1hier} for the hierarchical algorithm.

\begin{figure}[t]
\begin{center}
\begin{pspicture}(6,6)(0,0)
\rput{0}(3,3){\includegraphics[scale=0.5,trim = 20 20 10 0]{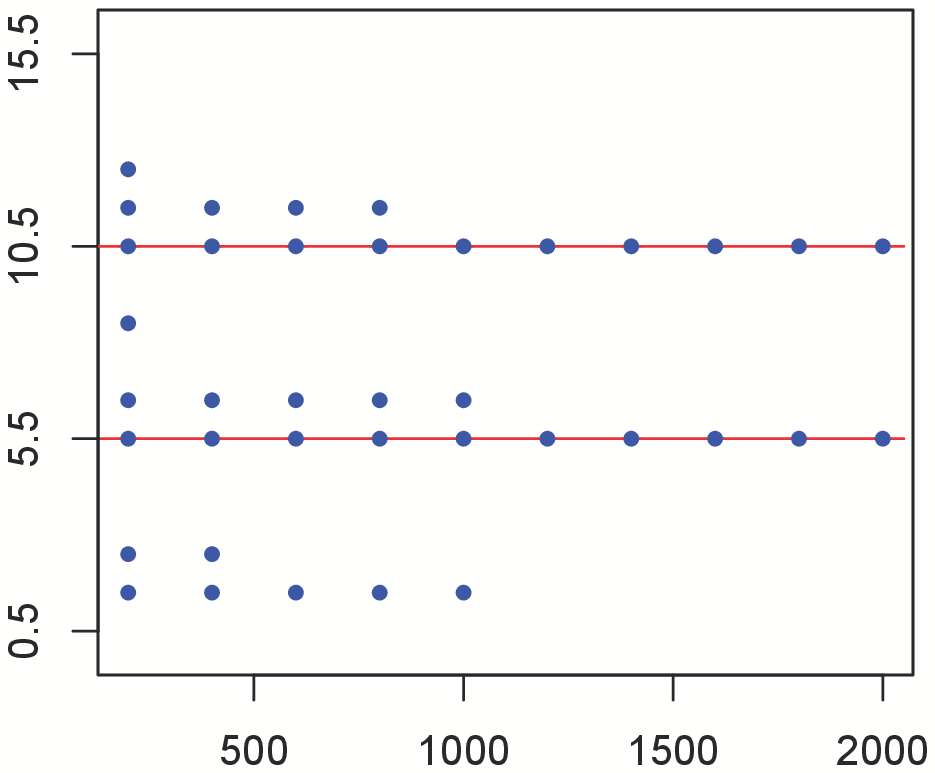}}
\rput[B]{0}(3.2,5.2){\scalebox{0.8}{$\p(I) = |A|^{|I|}-1$}}
\rput[B]{0}(3.45,0.3){\scalebox{0.8}{$n$}}
\rput[B]{90}(0.4,2.95){\scalebox{0.7}{$\hat U(\Xs)$}}
\end{pspicture}
\begin{pspicture}(6,6)(0,0)
\rput{0}(3,3){\includegraphics[scale=0.5,trim = 20 20 10 0]{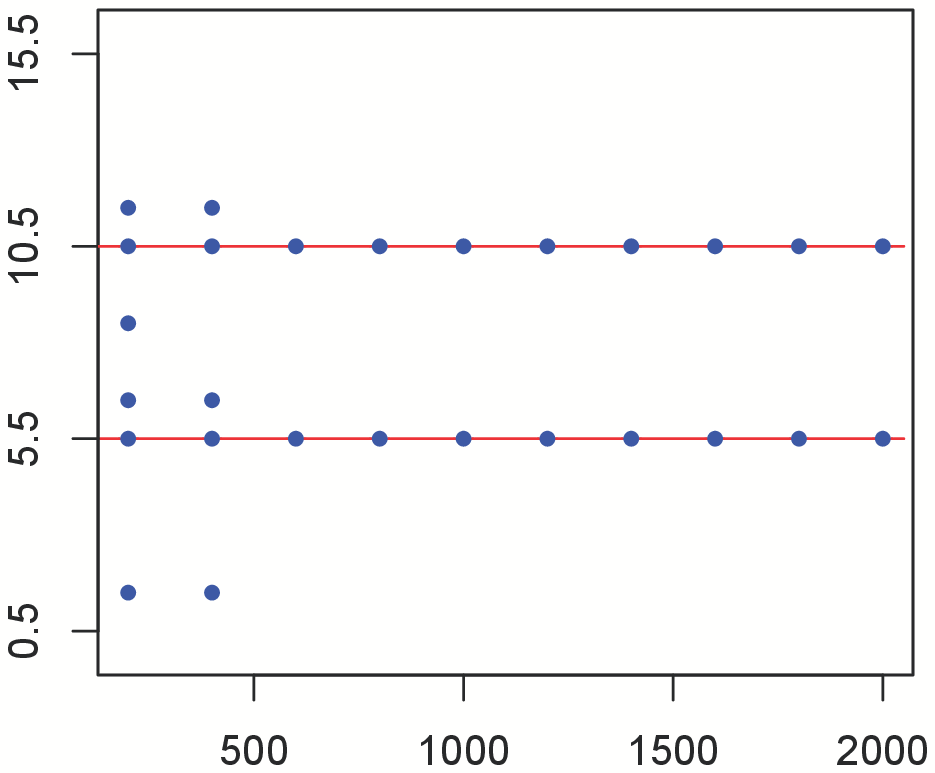}}
\rput[B]{0}(3.2,5.2){\scalebox{0.8}{$\p(I) = 0.5(|A|^{|I|}-1)$}}
\rput[B]{0}(3.45,0.3){\scalebox{0.8}{$n$}}
\rput[B]{90}(0.4,2.95){\scalebox{0.7}{$\hat U(\Xs)$}}
\end{pspicture}
\begin{pspicture}(6,6)(0,0)
\rput{0}(3,3){\includegraphics[scale=0.5,trim = 20 20 10 0]{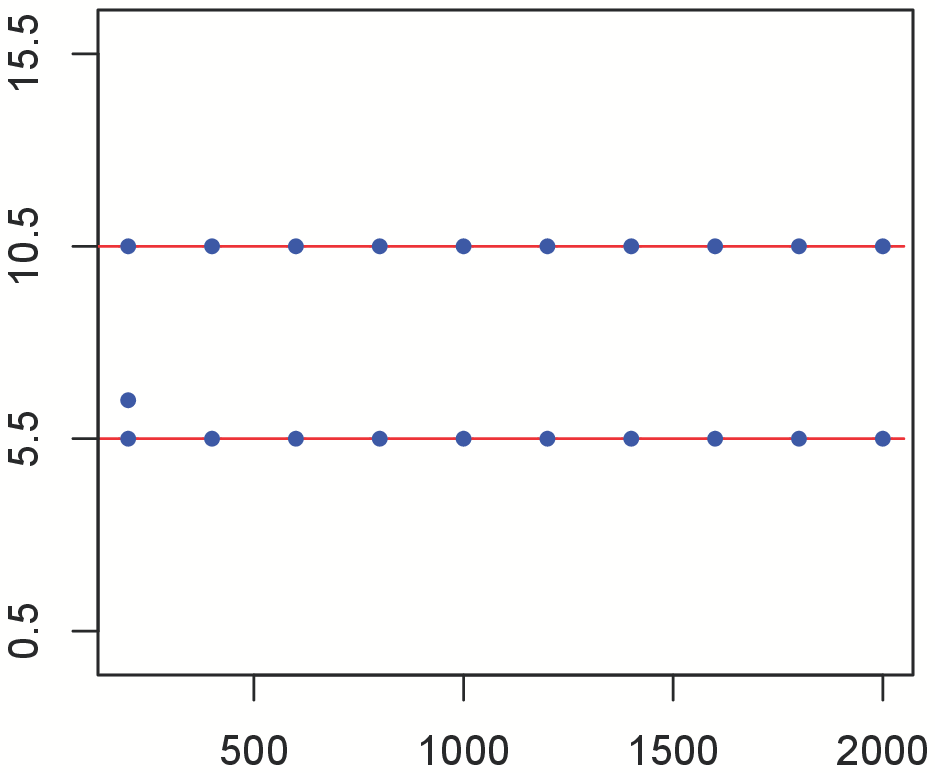}}
\rput[B]{0}(3.2,5.2){\scalebox{0.8}{$\p(I) = 0.1(|A|^{|I|}-1)$}}
\rput[B]{0}(3.45,0.3){\scalebox{0.8}{$n$}}
\rput[B]{90}(0.4,2.95){\scalebox{0.7}{$\hat U(\Xs)$}}
\end{pspicture}
\begin{pspicture}(6,6)(0,0)
\rput{0}(3,3){\includegraphics[scale=0.5,trim = 20 20 10 0]{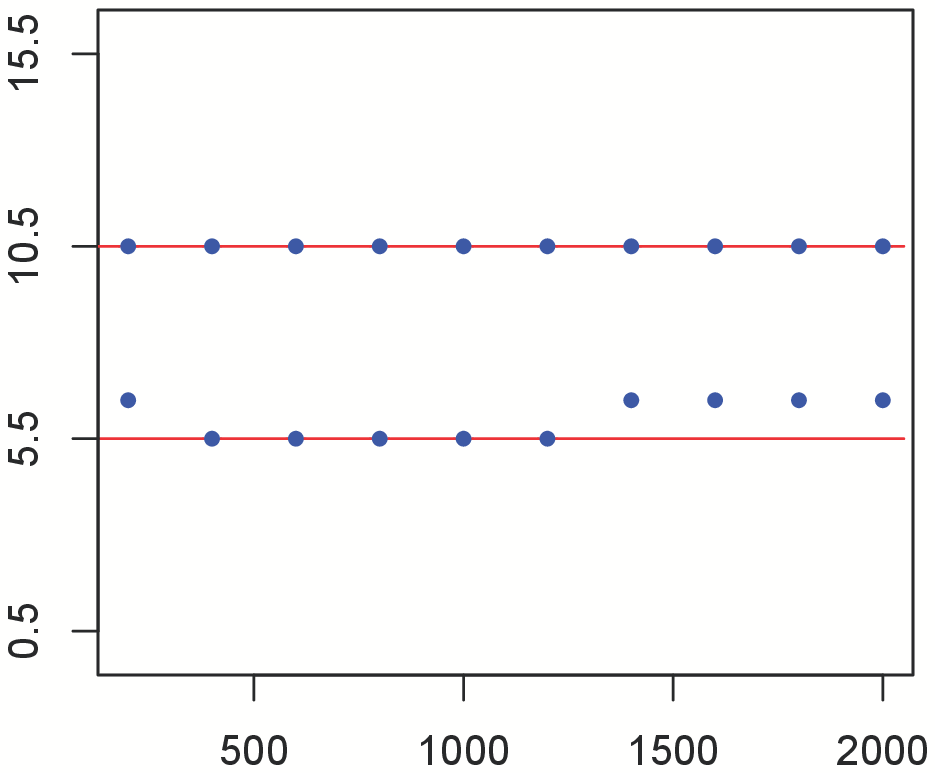}}
\rput[B]{0}(3.2,5.2){\scalebox{0.8}{$\p(I) = 0.05(|A|^{|I|}-1)$}}
\rput[B]{0}(3.45,0.3){\scalebox{0.8}{$n$}}
\rput[B]{90}(0.4,2.95){\scalebox{0.7}{$\hat U(\Xs)$}}
\end{pspicture}
\end{center}
\caption{Hierarchical algorithm estimator of the set of points of independence $U^*=\{5.5,10.5\}$, with $|A|=2$ and for different penalizing functions.}
\label{fig1hier}
\end{figure}

We performed another simulation study  for the alphabet $A=\{1,2,3\}$ and the transition matrices $P'_1$ and $P'_2$ given by
\[
P'_1 = \begin{pmatrix}
\frac13 & \frac13 & \frac13\\
0 & \frac23 & \frac13\\
\frac23 & 0 & \frac13
\end{pmatrix}
\,,\qquad P'_2 = \begin{pmatrix}
\frac12 & \frac12 & 0\\
\frac13 & \frac13 & \frac13\\
\frac16 & \frac56 & 0
\end{pmatrix}\,.
\]

We simulated as before $n$ realizations of the vector $\vec{X} = (X_1,X_2,\dotsc,X_{15})$
such that $X_1,\dotsc, X_5$ follow a markov chain with transition matrix $P'_1$, $X_6,\dotsc, X_{10}$ follow a markov chain with transition matrix $P'_2$
and $X_{11},\dotsc, X_{15}$ follow again a markov chain with transition matrix $P'_1$. The initial condition for any one of these markov chains is the uniform distribution over $A$. The results for different values of $n$ are given in Figure~\ref{fig2din} for the exact estimator and in  Figure~\ref{fig2hier} for the hierarchical algorithm. 

\begin{figure}[t]
\begin{center}
\begin{pspicture}(6,6)(0,0)
\rput{0}(3,3){\includegraphics[scale=0.5,trim = 20 20 10 0]{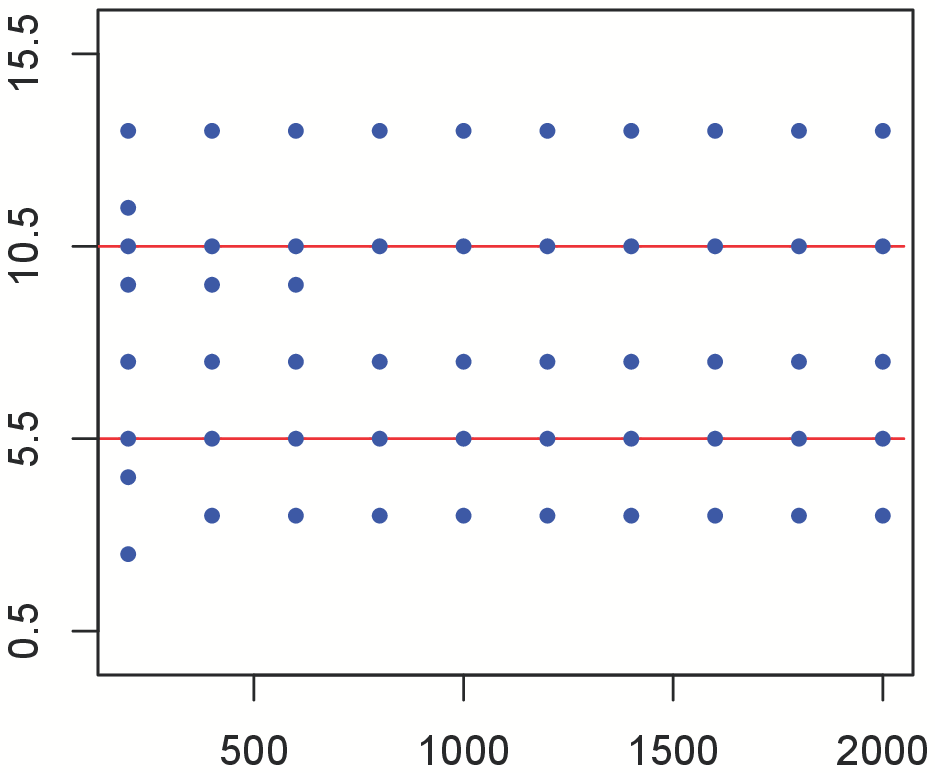}}
\rput[B]{0}(3.2,5.2){\scalebox{0.8}{$\p(I) = |A|^{|I|}-1$}}
\rput[B]{0}(3.45,0.3){\scalebox{0.8}{$n$}}
\rput[B]{90}(0.4,2.95){\scalebox{0.7}{$\hat U(\Xs)$}}
\end{pspicture}
\begin{pspicture}(6,6)(0,0)
\rput{0}(3,3){\includegraphics[scale=0.5,trim = 20 20 10 0]{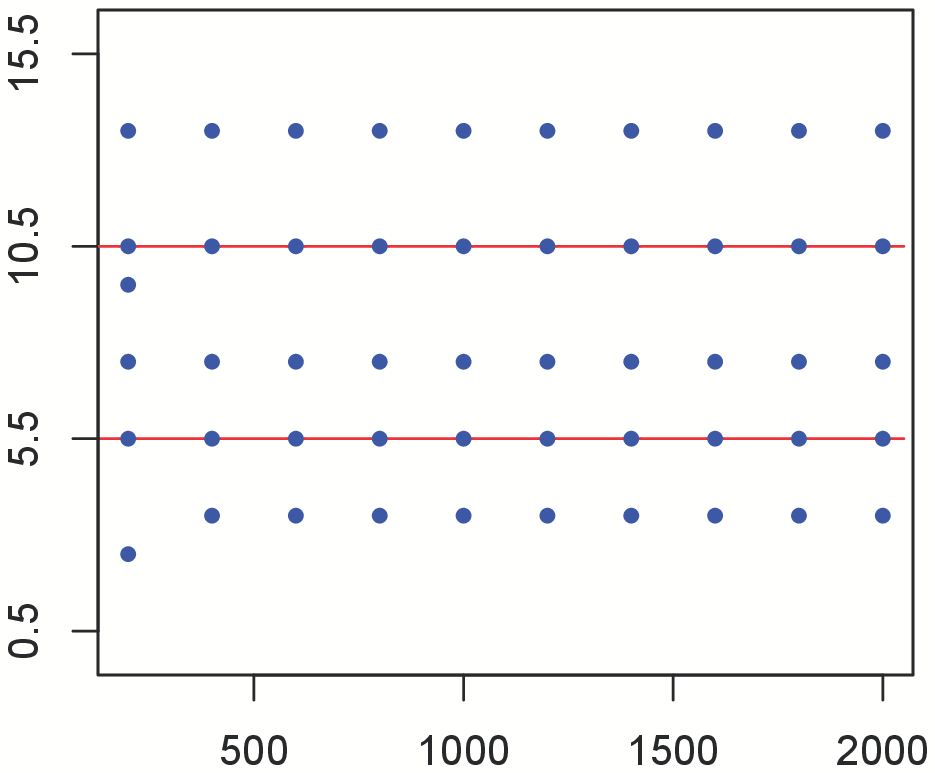}}
\rput[B]{0}(3.2,5.2){\scalebox{0.8}{$\p(I) = 0.5(|A|^{|I|}-1)$}}
\rput[B]{0}(3.45,0.3){\scalebox{0.8}{$n$}}
\rput[B]{90}(0.4,2.95){\scalebox{0.7}{$\hat U(\Xs)$}}
\end{pspicture}
\begin{pspicture}(6,6)(0,0)
\rput{0}(3,3){\includegraphics[scale=0.5,trim = 20 20 10 0]{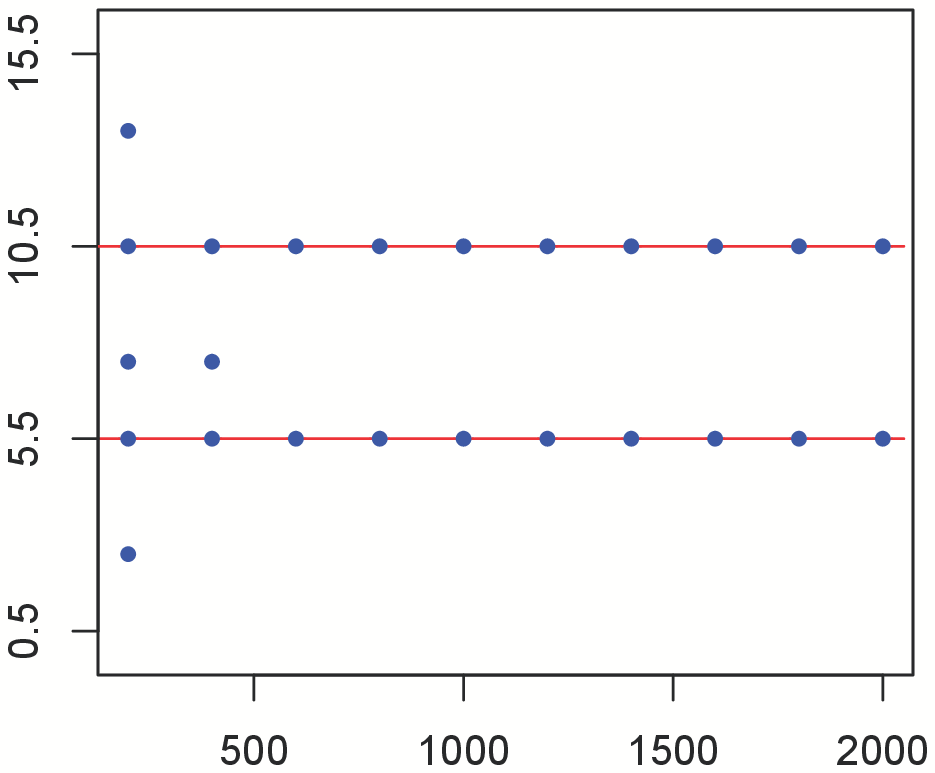}}
\rput[B]{0}(3.2,5.2){\scalebox{0.8}{$\p(I) = 0.1(|A|^{|I|}-1)$}}
\rput[B]{0}(3.45,0.3){\scalebox{0.8}{$n$}}
\rput[B]{90}(0.4,2.95){\scalebox{0.7}{$\hat U(\Xs)$}}
\end{pspicture}
\begin{pspicture}(6,6)(0,0)
\rput{0}(3,3){\includegraphics[scale=0.5,trim = 20 20 10 0]{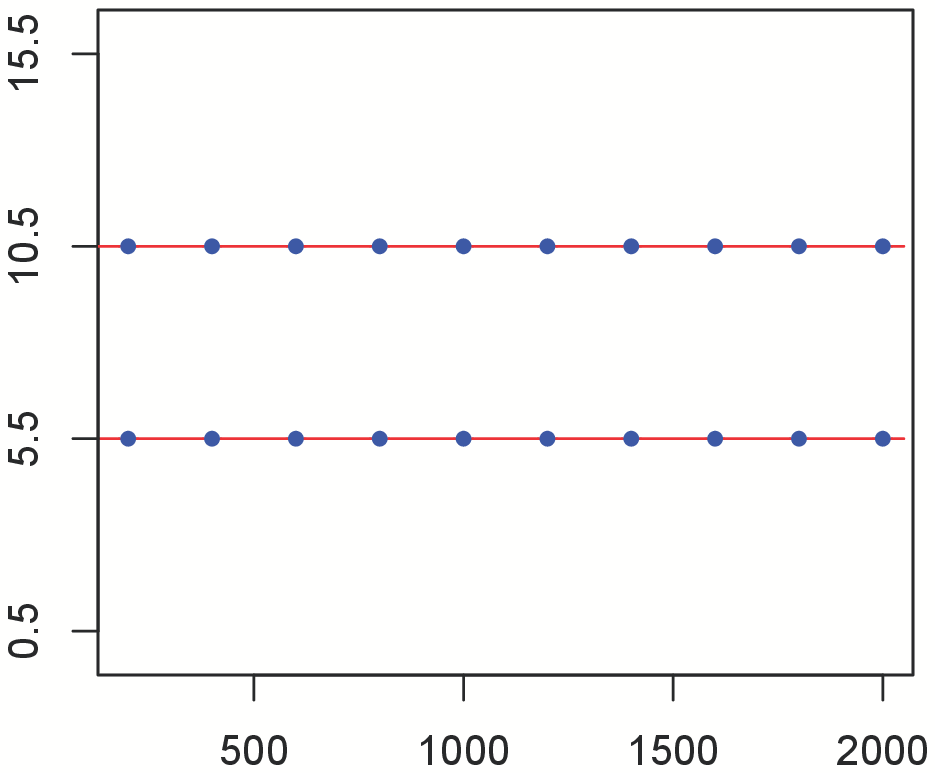}}
\rput[B]{0}(3.2,5.2){\scalebox{0.8}{$\p(I) = 0.05(|A|^{|I|}-1)$}}
\rput[B]{0}(3.45,0.3){\scalebox{0.8}{$n$}}
\rput[B]{90}(0.4,2.95){\scalebox{0.7}{$\hat U(\Xs)$}}
\end{pspicture}
\end{center}
\caption{Dynamic programming estimator of the set of points of independence $U^*=\{5.5,10.5\}$, with $|A|=3$  for different penalizing functions.}
\label{fig2din}
\end{figure}

\begin{figure}[t]
\begin{center}
\begin{pspicture}(6,6)(0,0)
\rput{0}(3,3){\includegraphics[scale=0.5,trim = 20 20 10 0]{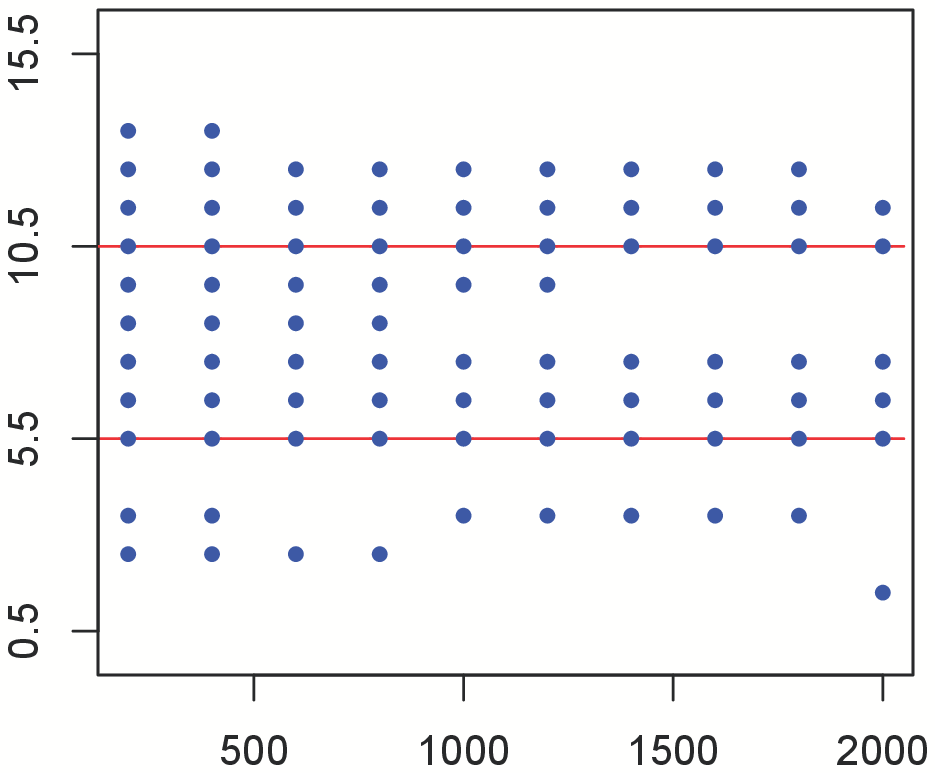}}
\rput[B]{0}(3.2,5.2){\scalebox{0.8}{$\p(I) = |A|^{|I|}-1$}}
\rput[B]{0}(3.45,0.3){\scalebox{0.8}{$n$}}
\rput[B]{90}(0.4,2.95){\scalebox{0.7}{$\hat U(\Xs)$}}
\end{pspicture}
\begin{pspicture}(6,6)(0,0)
\rput{0}(3,3){\includegraphics[scale=0.5,trim = 20 20 10 0]{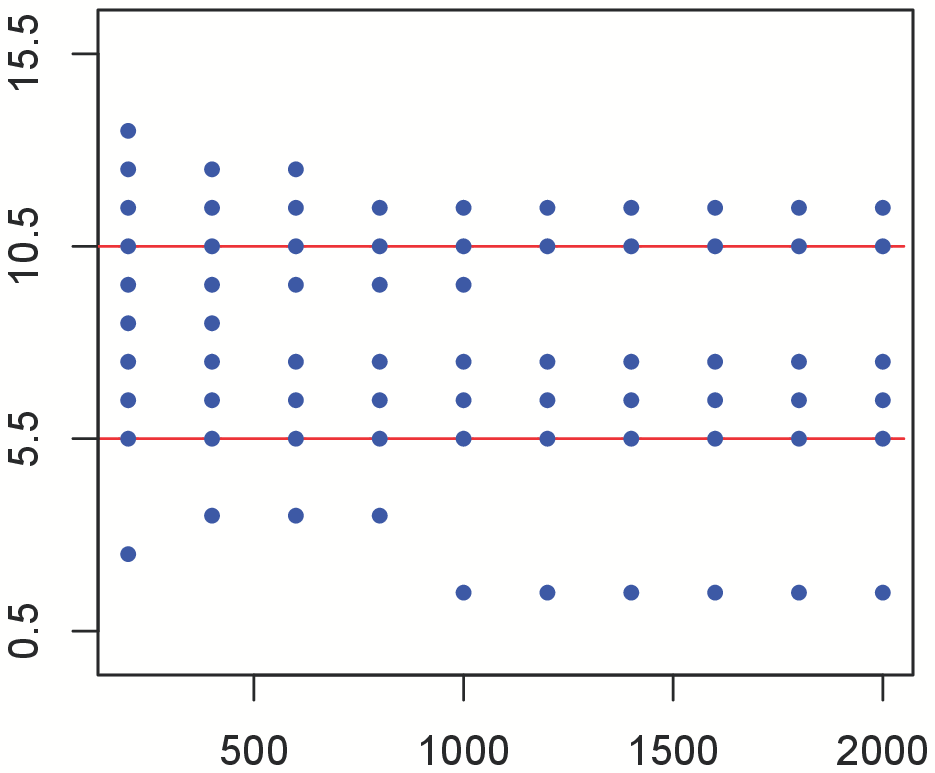}}
\rput[B]{0}(3.2,5.2){\scalebox{0.8}{$\p(I) = 0.5(|A|^{|I|}-1)$}}
\rput[B]{0}(3.45,0.3){\scalebox{0.8}{$n$}}
\rput[B]{90}(0.4,2.95){\scalebox{0.7}{$\hat U(\Xs)$}}
\end{pspicture}
\begin{pspicture}(6,6)(0,0)
\rput{0}(3,3){\includegraphics[scale=0.5,trim = 20 20 10 0]{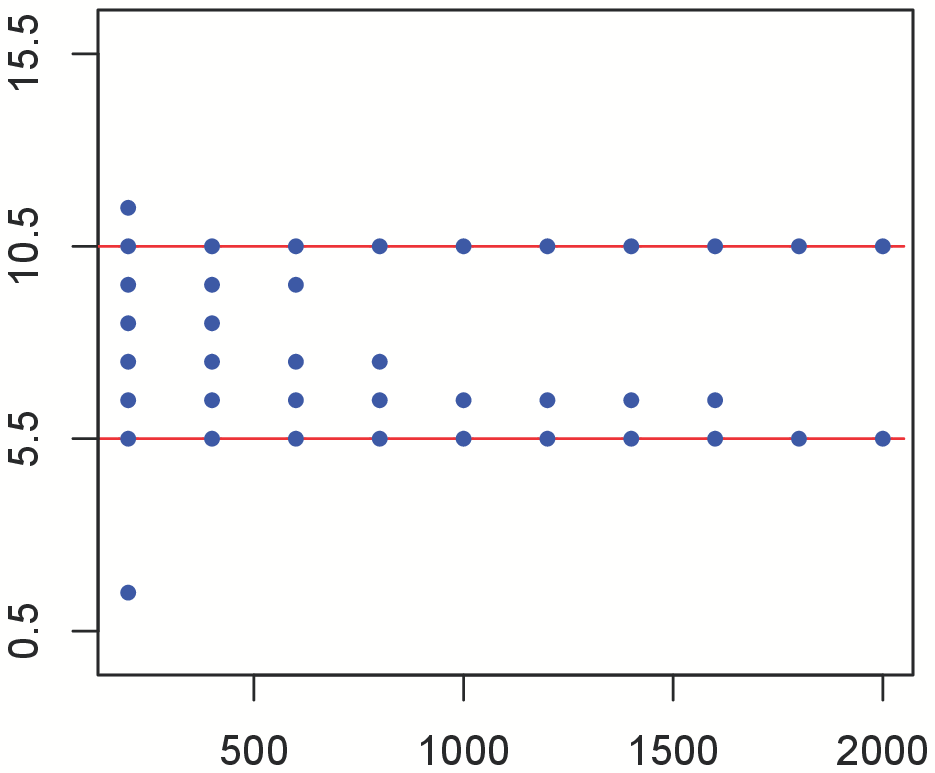}}
\rput[B]{0}(3.2,5.2){\scalebox{0.8}{$\p(I) = 0.1(|A|^{|I|}-1)$}}
\rput[B]{0}(3.45,0.3){\scalebox{0.8}{$n$}}
\rput[B]{90}(0.4,2.95){\scalebox{0.7}{$\hat U(\Xs)$}}
\end{pspicture}
\begin{pspicture}(6,6)(0,0)
\rput{0}(3,3){\includegraphics[scale=0.5,trim = 20 20 10 0]{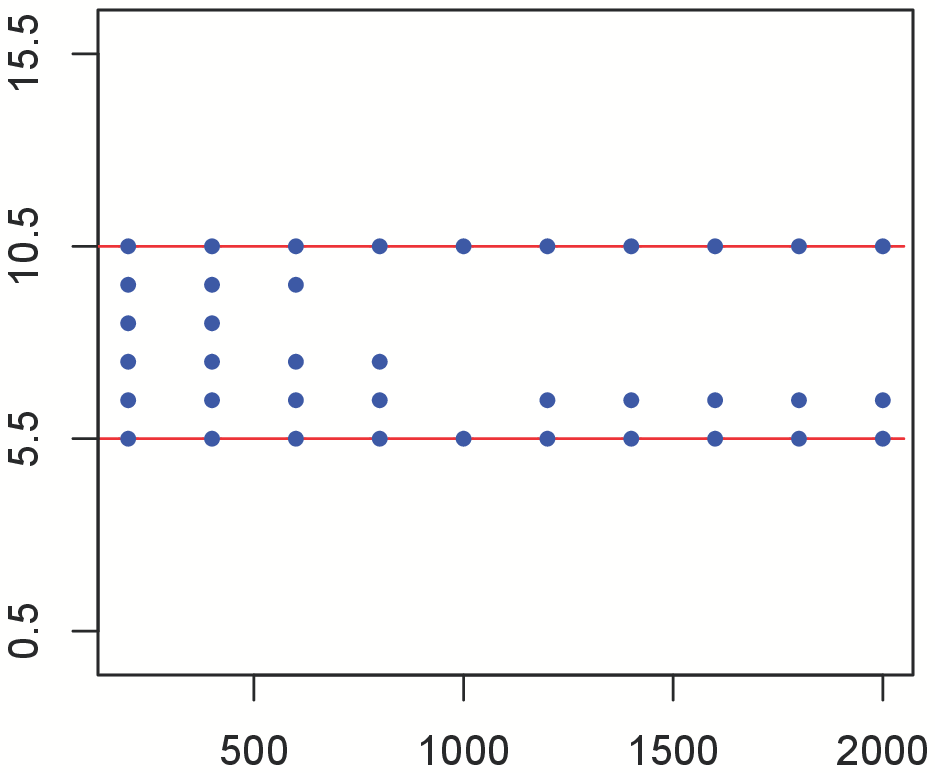}}
\rput[B]{0}(3.2,5.2){\scalebox{0.8}{$\p(I) = 0.05(|A|^{|I|}-1)$}}
\rput[B]{0}(3.45,0.3){\scalebox{0.8}{$n$}}
\rput[B]{90}(0.4,2.95){\scalebox{0.7}{$\hat U(\Xs)$}}
\end{pspicture}
\end{center}
\caption{Hierarchical algorithm estimator of the set of points of independence $U^*=\{5.5,10.5\}$, with $|A|=3$  for different penalizing functions.}
\label{fig2hier}
\end{figure}

The two algorithms were implemented in the open-source package {\tt R} and are available 
upon request.  

\section{Dimensionality reduction in Ebola Virus protein alignments}\label{data}

As an illustration of the practical use of the method proposed in this paper we 
apply the  dynamic programming algorithm described in Section~\ref{secalgs}
to a protein sequence alignment of Ebola Virus taken from the Ebola HFV Database \cite{kuikenetall2012}. 

The alignment corresponds to the matrix protein VP40, with $n=21$ and $m=326$. 
Although the possible amino-acids  in each position are 20, the maximal observed alphabet is
$A_{17}$ with 5 amino-acids. Moreover,  although the number of columns is much bigger than the number of observations, a total of  206 columns in the alignment shows a conserved amino-acid, with no influence in the segmentation result.  

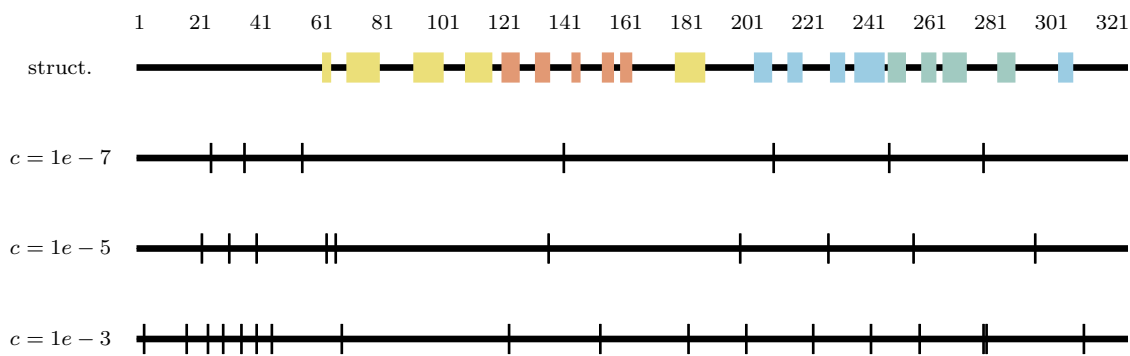
\begin{figure}[t]
\begin{center}
\psset{unit=0.4cm}
\begin{pspicture}(32.6,10)(-5,-3)
\psset{linewidth=2.5pt,arrowsize=12pt,linecolor=black}
\multido{\r=0.1+2,\i=1+20}{17}{\rput{0}(\r,9.5){\scriptsize \i}}
\definecolor{myyellow}{cmyk}{0,0,0.6,0.1}
\definecolor{myorange}{cmyk}{0,0.4,0.5,0.1}
\definecolor{myblue}{cmyk}{0.3,0,0,0.1}
\definecolor{mycyan}{cmyk}{0.3,0,0.2,0.1}
\psline(0,8)(32.6,8)\rput{0}(-2.5,8){\scriptsize struct.}
\psset{fillstyle=solid,fillcolor=myyellow,linecolor=myyellow}
\psframe(6.1,7.5)(6.4,8.5)
\psframe(6.9,7.5)(8,8.5)
\psframe(9.1,7.5)(10.1,8.5)
\psframe(10.8,7.5)(11.7,8.5)
\psframe(17.7,7.5)(18.7,8.5)
\psset{fillstyle=solid,fillcolor=myorange,linecolor=myorange}
\psframe(12,7.5)(12.6,8.5)
\psframe(13.1,7.5)(13.6,8.5)
\psframe(14.3,7.5)(14.6,8.5)
\psframe(15.3,7.5)(15.7,8.5)
\psframe(15.9,7.5)(16.3,8.5)
\psset{fillstyle=solid,fillcolor=myblue,linecolor=myblue}
\psframe(20.3,7.5)(20.9,8.5)
\psframe(21.4,7.5)(21.9,8.5)
\psframe(22.8,7.5)(23.3,8.5)
\psframe(23.6,7.5)(24.6,8.5)
\psframe(30.3,7.5)(30.8,8.5)
\psset{fillstyle=solid,fillcolor=mycyan,linecolor=mycyan}
\psframe(24.7,7.5)(25.3,8.5)
\psframe(25.8,7.5)(26.3,8.5)
\psframe(26.5,7.5)(27.3,8.5)
\psframe(28.3,7.5)(28.9,8.5)
\psset{linewidth=1pt,fillstyle=solid,fillcolor=black,linecolor=black}
\psline[linewidth=2.5pt](0,5)(32.6,5)\rput{0}(-2.5,5){\scriptsize$c=1e-7$}
\psline(2.45,4.5)(2.45,5.5)
\psline(3.55,4.5)(3.55,5.5)
\psline(5.45,4.5)(5.45,5.5)
\psline(14.05,4.5)(14.05,5.5)
\psline(20.95,4.5)(20.95,5.5)
\psline(24.75,4.5)(24.75,5.5)
\psline(27.85,4.5)(27.85,5.5)
\rput{0}(0,-3){
\psline[linewidth=2.5pt](0,5)(32.6,5)\rput{0}(-2.5,5){\scriptsize$c=1e-5$}
\psline(2.15,4.5)(2.15,5.5)
\psline(3.05,4.5)(3.05,5.5)
\psline(3.95,4.5)(3.95,5.5)
\psline(6.25,4.5)(6.25,5.5)
\psline(6.55,4.5)(6.55,5.5)
\psline(13.55,4.5)(13.55,5.5)
\psline(19.85,4.5)(19.85,5.5)
\psline(22.75,4.5)(22.75,5.5)
\psline(25.55,4.5)(25.55,5.5)
\psline(29.55,4.5)(29.55,5.5)}
\rput{0}(0,-6){
\psline[linewidth=2.5pt](0,5)(32.6,5)\rput{0}(-2.5,5){\scriptsize$c=1e-3$}
\psline(0.25,4.5)(0.25,5.5)
\psline(1.65,4.5)(1.65,5.5)
\psline(2.35,4.5)(2.35,5.5)
\psline(2.85,4.5)(2.85,5.5)
\psline(3.45,4.5)(3.45,5.5)
\psline(3.95,4.5)(3.95,5.5)
\psline(4.45,4.5)(4.45,5.5)
\psline(6.75,4.5)(6.75,5.5)
\psline(12.25,4.5)(12.25,5.5)
\psline(15.25,4.5)(15.25,5.5)
\psline(18.15,4.5)(18.15,5.5)
\psline(20.05,4.5)(20.05,5.5)
\psline(22.25,4.5)(22.25,5.5)
\psline(24.15,4.5)(24.15,5.5)
\psline(25.75,4.5)(25.75,5.5)
\psline(27.85,4.5)(27.85,5.5)
\psline(27.95,4.5)(27.95,5.5)
\psline(31.15,4.5)(31.15,5.5)}
\end{pspicture}
\caption{Sequence segmentation of Ebola virus protein VP40. The top represents the secondary structure 
of the protein described in  \cite{dessen2000}, while the other lines represent the segments estimated by the dynamic programming algorithm for different values of $c$.}\label{figebola}
\end{center}
\end{figure}

As the number of observed amino-acids differs considerably in each position, we use the empirical penalizing function given by $\hat\p(\Xs_I) = c(\max(2,\prod_{i\in I}\|\Xs_{i}\|_0)-1)$, where $\|\Xs_{i}\|_0 = 
\sum_{a_i\in A_i} \1\{N(a_i)>0\}$. 
In this case we obtained different segmentation patterns corresponding to different  penalizing functions, that can be compared to the secondary structure patterns in  \cite{dessen2000} (see Figure~\ref{figebola}).

\section*{Acknowledgments}
F.L.\/ is partially supported by a CNPq-Brazil fellowship (304836/2012-5) and FAPESP's fellowship 
(2014/00947-0). This article was produced as part of the activities of FAPESP  Research, Innovation and Dissemination Center for Neuromathematics, grant 2013/07699-0, S\~ao Paulo Research Foundation.

\bibliography{references}

\end{document}